\documentclass[11pt]{article}

\usepackage{amssymb}
\usepackage{epsfig}
\usepackage{amsmath}

 \usepackage{geometry}
 \newtheorem{theorem}{Theorem}[section]
 \newtheorem{lemma}{Lemma}[section]
 \newtheorem{corollary}{Corollary}[section]
 \newtheorem{claim}{Claim}
 \newtheorem{definition}{Definition}[section]

 \newcommand{\qed}{\vrule height4pt width4pt depth2pt}

\newcounter{algleo}
\newlength{\lefttab}
\newlength{\numberoffset}
\newenvironment{algleo}%
  {\trivlist
   \topsep=0pt\parsep=0pt\itemsep=0pt
   \def\li{\item\refstepcounter{algleo}\makebox[0.8em][r]{\thealgleo\hspace{\numberoffset}}
       \hangafter1\hangindent1.8em\noindent}%
   \def\linonumber{\item\makebox[0.8em][r]{\hspace{\numberoffset}}
       \hangafter1\hangindent1.8em\noindent}%
   \addtolength{\lefttab}{1.25em}
   \addtolength{\numberoffset}{1.25em}
   \leftskip=\lefttab}%
  {\endtrivlist}

\begin{document}

\title{Algorithmic Aspects of Homophyly of Networks}

\author{
Angsheng Li
\thanks{Institute of Software, Chinese Academy of Sciences, China.
  Email: \texttt{angsheng@ios.ac.cn}.
  The author is supported by the hundred talent program of the Chinese Academy of Sciences,
  and the grand challenge program, {\it Network Algorithms and Digital Information},
  Institute of Software, Chinese Academy of Sciences.}
\and
Peng Zhang
\thanks{Corresponding author.
  School of Computer Science and Technology, Shandong University, China
  Email: \texttt{algzhang@sdu.edu.cn}.
  The author is supported by the National Natural Science Foundation of China (60970003),
  the Special Foundation of Shandong Province Postdoctoral Innovation Project (200901010),
  and the Independent Innovation Foundation of Shandong University (2012TS072).}
}

\maketitle

\begin{abstract}
We investigate the algorithmic problems of the {\it homophyly
phenomenon} in networks. Given an undirected graph $G = (V, E)$ and a vertex
coloring $c \colon V \rightarrow \{1, 2, \cdots, k\}$ of $G$, we say that
a vertex $v\in V$ is {\it happy} if $v$ shares the same color with all its
neighbors, and {\it unhappy}, otherwise, and that an edge $e\in E$ is
{\it happy}, if its two endpoints have the same color, and {\it unhappy},
otherwise. Supposing $c$ is a {\it partial vertex coloring} of $G$,
we define the Maximum Happy Vertices problem (MHV, for
short) as to color all the remaining vertices such that the number
of happy vertices is maximized, and the Maximum Happy Edges problem
(MHE, for short) as to color all the remaining vertices such that
the number of happy edges is maximized.

Let $k$ be the number of colors allowed in the problems.
We show that both MHV and MHE can be solved in polynomial time
if $k = 2$, and that both MHV and MHE are NP-hard if $k \geq 3$.
We devise a $\max \{1/k, \Omega(\Delta^{-3})\}$-approximation algorithm
for the MHV problem, where $\Delta$ is the maximum degree of vertices in
the input graph, and a $1/2$-approximation algorithm for the MHE problem.
This is the first theoretical progress of these two natural and fundamental
new problems.
\end{abstract}

\section{Introduction}
\label{sec - introduction}
Networks or at least social networks heavily depend on human or
social behaviors. It is believed that {\em homophyly}~\cite[Chapter
4]{EK10} is one of the most basic notions governing the structure of
social networks. It is a common sense principle that people are more
likely to connect with people they like, as what says in the proverb
``birds of a feather flock together''.

Li and Peng in \cite{LP11,LP12} gave a mathematical definition of
{\it community}, and {\it small community phenomenon} of networks,
and showed that networks from some classic models do satisfy the
small community phenomenon. A. Li and J. Li et al. \cite{LLPP11} proposed
a homophyly model by introducing a color for every vertex in the classical
preferential attachment
model such that networks generated from this model satisfy
simultaneously the following properties: 1) power law degree
distribution, 2) small diameter property, 3) vertices of the same color
naturally form a small community, and 4) almost all vertices are
contained in some small communities, i.e., the small community
phenomenon of networks. This result implies the {\it homophyly law}
of networks that the mechanism of the small community phenomenon is
homophyly, and that vertices within a small community share
remarkable common features.

A. Li and J. Li et al. \cite{LLPP12} showed that many real networks
satisfy exactly the homophyly law, in which an interesting application is
the prediction and confirmation of keywords from a paper citation network
of high energy physics
theory\footnote{\texttt{http://snap.stanford.edu/data/cit-HepTh.html}.}.
The network contains $27,770$ vertices (i.e., papers) and $352,807$ edges
(i.e., citations). All the papers have titles and abstracts, but only $1,214$
papers have keywords listed by their authors.
We interpret the keywords of a paper to be a {\it function} of the paper.
By the homophyly law, vertices within a small community of the network must
share remarkable common features (keywords here).
The prediction is as follows: 1) to find a
small community from each vertex, if any, 2) to extract
the most popular $5$ keywords from the known keywords in a community,
as the remarkable common features of this community,
3) to predict that (all or part of) the $5$ remarkable common
keywords are keywords of a paper in the community, 4) to confirm a
prediction of keyword $K$ for a paper $P$, if $K$ appears in either
the title or the abstract of paper $P$. It is a surprising result
that this simple prediction confirms keywords for $19,200$ papers in
the network. This experiment implies that real networks do satisfy
the homophyly law, and that the homophyly law is the principle for
prediction in networks.

The keywords can be viewed as the attributes of vertices in a network.
The above experimental result suggests a natural theoretical problem that,
given a network in which some vertices have their attributes unfixed,
how to assign attributes to these vertices such that the resulting
network reflects the homophyly law in the most degree?
Some attributes of a vertex cannot be changed, such as nationality, sex,
color and language, but some other attributes can be changed,
such as interest, job, income and working place.
For simplicity, we consider the case that each vertex contains only one
alterable attribute, i.e., the network is a $1$-dimensional network.
Consider the following scenario. Suppose in a company there are many employees
which constitutes a friendship network.
Some employees have been assigned to work in some departments of the
company, while the remaining employees are waiting to be assigned.
An employee is {\em happy}, if s/he works in the same department with all of
(or $\rho$ fraction of for some $\rho \in (0, 1]$, or at least $q$ for some
integer $q > 0$) her/his friends; otherwise s/he is {\em unhappy}.
Similarly, a friendship is {\em happy} (or lucky) if the two
related friends work in the same department; otherwise the friendship is
{\em unhappy}. Our goal is to achieve the greatest social benefits, that is,
to maximize the number of {\em happy vertices} (similarly, {\em happy edges})
in the network.



We can easily express the above problems as graph coloring problems,
just identifying each attribute value with a different color.
Hence we get two specific graph coloring problems, as defined below.

\begin{definition}[The MHV problem]
(Instance) In the Maximum Happy Vertices (MHV) problem, we are given
an undirected graph $G = (V, E)$, a color set $C = \{1, 2, \cdots, k \}$,
and a partial vertex coloring function $c \colon V \to C$. We say that $c$
is a partial function in the sense that $c$ assigns colors to part of
vertices in $V$.

(Query) A vertex is {\em happy} if it shares the same color with all its
neighbors, otherwise it is {\em unhappy}. The task is to extend $c$ to a total
function $c'$ such that the number of happy vertices is maximized.
\end{definition}

\begin{definition}[The MHE problem]
(Instance) The input of the Maximum Happy Edges (MHE) problem is the same
as that of the MHV problem.

(Query) An edge is {\em happy} if its two endpoints have the same color,
otherwise it is {\em unhappy}. The goal is to extend $c$ to a total function $c'$
such that the number of happy edges is maximized.
\end{definition}

The vertex coloring defined by the total function $c' \colon V \rightarrow C$
in MHV and MHE is called a {\em total vertex coloring}.
In general, a (partial or total) vertex coloring can be denoted by
$(V_1, V_2, \cdots, V_k)$, where $V_i$ is the set of all vertices having
color $i$. A total vertex coloring is a partition of $V(G)$, while a partial
vertex coloring may not. Therefore, the MHV and MHE problems are two extension
problems from a partial vertex coloring to a total vertex coloring.
We remark that the coloring for our case is completely different from the
well-known Graph Coloring problem, which requires that the two endpoints of
an edge must be colored differently and asks to color a graph in such a way
by using the minimized number of colors. We use the notion of color just
for intuition.

If in the MHV problem the color number $k$ is a constant, the problem is
denoted by $k$-MHV. For the specific values of $k$, we have the 2-MHV problem,
the 3-MHV problem, and so on. Note that in the original MHV problem $k$ is
given as a part of the input. Similarly, we have the $k$-MHE problem for
constant $k$, with 2-MHE, 3-MHE, etc. being its specific problems.

We remark that both the MHV and MHE problems are natural and fundamental
algorithmic problems, and that they have not appeared yet in literature.
The reasons could be two folds. On the one hand, we ask the questions from
our network applications which did not happen before; on the other hand,
the meaning of coloring has been specified previously so that the two
endpoints of an edge must have different colors.
%
%
We notice that the current version of our problems may not really help
network applications much because of their simplicity.
For real network applications, probably the experimental method \cite{LLPP12}
introduced at the beginning of this section is fine enough.
However, this has no theoretical guarantee, owing to different structures
of networks. Our problems seem essentially new and fundamental algorithmic
problems. Theoretical analysis of the problems are always helpful to
understand the nature of the problems, and hence are very welcome.

\subsection{Our Results}
We investigate algorithms to solve the MHV and MHE problems.
It is easy to see that the partial function $c$ plays an important role
in the MHV and MHE problems. If none of the vertices in the input graph has
a pre-specified color, then the MHV and MHE problems are trivial. The optimal
solution just assigns one arbitrary color to all the vertices. This will
make all vertices and all edges happy.

We prove that the MHV and MHE problems are NP-hard. Interestingly, the
complexity of $k$-MHV and $k$-MHE dramatically changes when $k$ changes
from 2 to 3. Specifically, we prove that both 2-MHV and 2-MHE can be solved
in polynomial time, while both $k$-MHV and $k$-MHE are actually NP-hard for
any constant $k \geq 3$. We thus seek approximation algorithms for the MHV
and MHE problems, and their variants $k$-MHV and $k$-MHE ($k \geq 3$).

We design two approximation algorithms
{\sc Greedy-MHV} (Subsection \ref{subsec - greedy approxalg for MHV}) and
{\sc Growth-MHV} (Subsection \ref{subsec - subset-growth approxalg for MHV})
for the MHV problem and its variant $k$-MHV.
Algorithm {\sc Greedy-MHV} is a simple greedy algorithm with approximation
ratio $1/k$. Algorithm {\sc Growth-MHV} is an algorithm based on the
subset-growth technique with approximation ratio $\Omega(\Delta^{-3})$,
where $\Delta$ is the maximum degree of vertices in the input graph.
In real networks, $\Delta$ is usually $\mbox{poly} \log n$, implying that
the ratio $\Omega(\Delta^{-3})$ is reasonable.
As Algorithm {\sc Growth-MHV} is executing, more and more vertices are
colored. According to the current vertex coloring for the input graph,
we define several types for the vertices. (Note that the types here are not
colors.) Algorithm {\sc Growth-MHV} works based on carefully classifying
all the vertices into several types.

We can extend our algorithms for MHV to deal with two more natural variants
SoftMHV and HardMHV.
In the SoftMHV problem, a vertex $v$ is happy if $v$ shares the same color with
at least $\rho \deg(v)$ neighbors, where $\rho$ (that is, the soft threshold)
is a number in $(0, 1]$ and $\deg(v)$ is the degree of vertex $v$.
In the HardMHV problem, a vertex $v$ is happy if $v$ shares the same color with
at least $q$ neighbors, where $q$ (that is, the hard threshold) is an integer.
We show that the SoftMHV and HardMHV problems can also be approximated within
$\max \{1/k, \Omega(\Delta^{-3})\}$.
The approximation algorithms for SoftMHV and HardMHV, given in the Appendix
for completeness, are similar to that for MHV.

For the MHE problem and its variant $k$-MHE, we devise a simple
approximation algorithm based on a division strategy, namely,
Algorithm {\sc Division-MHE} (Section \ref{sec - algorithms for MHE}).
The approximation ratio is proved to be 1/2.

\subsection{Related Work and Relation to Other Problems}
\label{sec - related works}
The MHV and MHE problems are two quiet natural vertex classification problems
arising from the homophyly phenomenon in networks.
Classification is a fundamental problem and has wide applications in
statistics, pattern recognition, machine learning, and many other fields.
Given a set of objects to be classified and a set of colors, a classification
problem can be depicted as from a very high level assigning a color to
each object in a way that is consistent with some observed data or structure
that we have about the problem \cite{BFOS84,KT02}.
In our problems, the observed strucute is homophyly.
Since the MHV and MHE problems are essentially new, in the following we just
show some closely related problems and results.

Thomas Schelling \cite{S72,S78}, the Nobel economics prize winner, showed
by experiments how global patterns of spatial segregation arise from the
effect of homophyly operating at the local level.
The experiments in \cite{S72} are given in one-dimensional and two-dimensional
geometric models.
From a more general viewpoint of graph theory, Schelling's experiments,
although given in geometric models, can be viewed as how to remove and add
edges from/to a graph whose vertices are all colored by some colors
such that the resulting graph possesses the homophyly property.
In contrast, the MHV and MHE problems are how to color the vertices
in a given graph whose part of vertices are already colored such that
the resulting graph possesses the homophyly property.

The Multiway Cut problem \cite{EL92,DJP+94,CKR00,KKS+04} should be
the traditional optimization problem that is most related to MHV and MHE.
Given an undirected graph $G = (V, E)$ with costs defined on edges and
a terminal set $S \subseteq V$, the Multiway Cut problem asks for a set
of edges (called a {\em multiway cut}, or simply a {\em cut}) with the minimum
total cost such that its removal from graph $G$ separates all terminals in $S$
from one another. The Multiway Cut problem in general graphs is NP-hard even
the terminal set contains only three terminals and each edge has a unit
cost \cite{DJP+94}. The current best approximation ratio known for this
problem is 1.3438 \cite{KKS+04}.

Removing a minimum multiway cut from a graph breaks the graph into several
components such that each component contains exactly one terminal.
From the viewpoint of graph coloring, this is equivalent to coloring the
uncolored vertices in a graph in which each terminal has a distinct
pre-specified color, such that the number of happy edges is maximized.
Therefore, the MHE problem is actually the dual of the Multiway Cut problem.
See Figure \ref{fig - multiway cut and vertex coloring} for an example.
(More precisely, the dual of Multiway Cut is only a special case of MHE,
since in MHE there may be more than one vertices having the same pre-specified
color.) However, Multiway Cut and MHE are quite different in terms of
approximation, since one is a maximization problem while the other is
a minimization problem.

\begin{figure}
\begin{center}
\includegraphics*[width=0.45\textwidth]{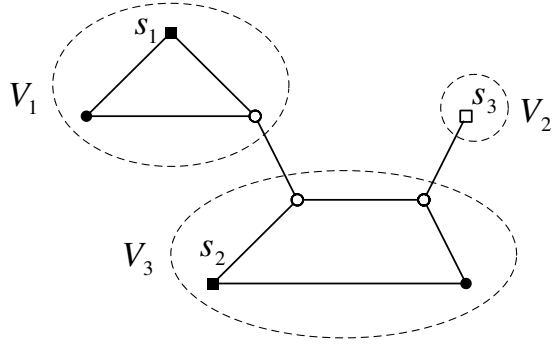}
\end{center}
\caption{An instance of Multiway Cut and the induced vertex coloring.
The square vertices are terminals and have pre-specified colors,
while the round vertices are non-terminal vertices.
The hollow vertices are border vertices.}
\label{fig - multiway cut and vertex coloring}
\end{figure}

For a vertex subset $V' \subseteq V(G)$ of graph $G$, we define the {\em border}
of $V'$ to be the set of vertices in $V'$ that has a neighbor not in $V'$.
Given a vertex coloring $(V_1, V_2, \cdots, V_k)$ of graph $G$,
the vertices in the border of each $V_i$ are obviously unhappy.
The MHV problem, which finds a vertex coloring that maximizes the number of
happy vertices, is actually equivalent to finding a vertex coloring
$(V_1, V_2, \cdots, V_k)$ for a graph in which some vertices are already
colored, such that the total number of vertices in borders of all $V_i$'s
is minimized. Please refer to Figure \ref{fig - multiway cut and vertex coloring}
for an example. The latter problem we just introduce is a new minimization
problem; the MHV problem and this new problem are dual to each other.

From the above analysis, one can see that the partial function $c$ in
the MHE problem (and the MHV problem), which assigns colors to part of
vertices of the input graph, actually simulates and generalizes
the {\em terminal set} part in the Multiway Cut problem.

Kann and Khanna et al. \cite{KKL+97} studied the Max $k$-Cut problem \cite{FJ97}
and its dual, that is, the Min $k$-Partition problem \cite{KKL+97}.
Given an undirected graph $G = (V, E)$, the Min $k$-Partition problem asks
to find a vertex coloring $c \colon V \rightarrow \{1, 2, \cdots, k\}$
such that the number of edges whose two endpoints have the same color
(i.e., the happy edges in our setting) is minimized.

According to the way of definitions in \cite{KKL+97}, we can define
the dual of the Min $k$-Cut problem \cite{SV95} as follows:
Given an undirected graph $G = (V, E)$ and an integer $k > 0$,
finding an edge subset whose removal breaks graph $G$ into {\em exactly}
$k$ components, such that the number of remaining edges is maximized.
Let's call this problem the Max $k$-Partition problem.
In other words, Max $k$-Partition asks for a total vertex coloring
$c' \colon V \rightarrow \{1, 2, \cdots, k\}$ such that the number of
happy edges is maximized, where $c'$ should be a surjective function
(that is, for each color $i$ there exists a vertex whose color is $i$).

The Max $k$-Partition problem defined as above is close to the MHE problem,
but they are still different in the obvious way:
In Max $k$-Partition there is no any vertex having a pre-specified
color and the required vertex coloring function $c'$ must be surjective,
while in MHE there must be some vertices having the pre-specified colors
and the required vertex coloring function $c'$ may not be surjective.

\bigskip
{\bf Notations.}
Let $G = (V, E)$ be a graph. Let $n = |V|$ and $m = |E|$.
Suppose $v \in V$ is a vertex. Denote by $N(v)$ the set of neighbors of $v$.
As usual, $\deg(v)$ means the degree of $v$, i.e, $\deg(v) = |N(v)|$.
Denote by $N^2(v)$ the set of neighbors of neighbors of $v$ (not including
$v$ itself), i.e., the vertices within distance 2 of $v$ (assume each edge
has unit distance).

Given a vertex coloring $c$, for a (colored or uncolored) vertex $v$, define
$N^u(v)$ as the set of vertices in $N(v)$ that has not yet been colored.
For a colored vertex $v$, define $N^s(v)$ as the set of vertices in $N(v)$
having the {\em same} color as $c(v)$, $N^d(v)$ as the set of vertices in
$N(v)$ having colors {\em different} to $c(v)$.

Given an instance $\cal I$ of some optimization problem $\cal P$,
we use $OPT({\cal I})$ ($OPT$ for short) to denote the optimum (that is,
the value of an optimal solution) of the instance.
Let $\cal A$ be an algorithm for problem $\cal P$.
We use $SOL({\cal I})$ ($SOL$ for short) to denote the value of the solution
found by algorithm $\cal A$ on instance $\cal I$ of problem $\cal P$.
In addition, $OPT$ and $SOL$ also denote the corresponding solutions,
abusing notations slightly.

\bigskip
{\bf Organization of the paper.}
The remaining of the paper is organized as follows.
In Section \ref{sec - algorithms for MHV}, we show that 2-MHV is
polynomial-time solvable, and give the greedy approximation algorithm
and the subset-growth approximation algorithm for the MHV and $k$-MHV
($k \geq 3$) problems.
In Section \ref{sec - algorithms for MHE}, we show that 2-MHE is
polynomial-time solvable, and give the division-strategy based approximation
algorithm for the MHE $k$-MHE ($k \geq 3$) problems.
In Section \ref{sec - Hardness results}, we prove the NP-hardness for the MHE,
$k$-MHE ($k \geq 3$), MHV, and $k$-MHV ($k \geq 3$) problems.
In Section \ref{sec - conclusions} we conclude the paper by
introducing some future work. In the Appendix, we give approximation
algorithms for the SoftMHV and HardMHV problems.

\section{Algorithms for MHV}
\label{sec - algorithms for MHV}
In Subsection \ref{subsec - 2-MHV Is in P}, we give the polynomial time exact
algorithm for the 2-MHV problem. In Subsection \ref{subsec - Approxalgs for MHV},
we give the approximation algorithms {\sc Greedy-MHV} and {\sc Growth-MHV} for
the MHV problem.

\subsection{2-MHV Is in P}
\label{subsec - 2-MHV Is in P}
Let $U$ be a finite set. Recall that a function
$f \colon 2^U \rightarrow \mathbf{Z}^+$ is said to be submodular if
$f(X) + f(Y) \geq f(X \cup Y) + f(X \cap Y)$ holds for all $X, Y \subseteq U$.
Given a vertex subset $V' \subseteq V(G)$, define function $f(V')$ to be
the number of vertices in $V'$ that has neighbors outside of $V'$, i.e.,
$f(V')$ is the size of the border (see Subsection \ref{sec - related works})
of $V'$. It is easy to verify that $f$ is a submodular function.

Consider the 2-MHV problem, in which the color set $C$ contains only two
colors 1 and 2. This problem can be solved in polynomial
time.

\begin{theorem}
The 2-MHV problem can be solved in $O(mn^7\log n)$ time.
\end{theorem}
\begin{proof}
Let $V_1^{org}$ be the set of vertices
that are colored by color 1 by the partial function $c$, and $V_2^{org}$
be the analogous vertex subset corresponding to color 2.
Then the 2-MHV problem is equivalent to finding a cut $(V_1, V_2)$
such that $V_i^{org} \subseteq V_i$ for $i=1, 2$ and $f(V_1) + f(V_2)$
is minimized. We can do this by merging all vertices in $V_1^{org}$ to
a single vertex $s$, all vertices in $V_2^{org}$ to a single vertex $t$,
and finding an $s$-$t$ cut $(V_1, V_2)$ on the resulting graph such that
$f(V_1) + f(V_2)$ is minimized.
As pointed out by \cite[Lemma 3]{ZNI05}, finding such a cut can be done
by an algorithm in \cite{IFF01} for minimizing submodular functions
in  $O(\theta n^7 \log n)$ time, where $\theta$ is the time to compute
the submodular function $f$.
When the input graph is stored by a collection of adjacency lists,
$f(\cdot)$ can be computed in $O(m)$ time in a straightforward way
(assuming the input graph contains no isolated vertex).
The proof of the theorem is finished.
\qed
\end{proof}

\subsection{Approximation Algorithms for MHV}
\label{subsec - Approxalgs for MHV}
The approximation algorithms for MHV work based on the types defined for
vertices, as shown in Definition \ref{def - types of vertices in MHV}.

\begin{definition}[Types of vertices in MHV]
\label{def - types of vertices in MHV}
Fix a (partial or total) vertex coloring. Let $v$ be a vertex. Then,
\begin{enumerate}
\item $v$ is an {\em $H$-vertex} if $v$ is colored and happy (i.e., $|N^s(v)| = \deg(v)$);
\item $v$ is a {\em $U$-vertex} if $v$ is colored and destined to be unhappy
(i.e., $|N^d(v)| > 0$);
\item $v$ is a {\em $P$-vertex} if
\begin{enumerate}
    \item $v$ is colored,
    \item $v$ has not been happy (i.e., $|N^s(v)| < \deg(v)$), and
    \item $v$ may become happy in the future (i.e., $|N^d(v)| = 0$);
\end{enumerate}
\item $v$ is an {\em $L$-vertex} if $v$ has not been colored.
\end{enumerate}
\end{definition}

See Figures \ref{fig - process a P-vertex}, \ref{fig - process a Lh-vertex},
\ref{fig - process a Lu-vertex} for examples of the vertex types.
Note that by a type name we also mean the set of vertices of that type.
Conversely, by a set name we also mean that each element in the set
is of that type. For example, $H$ is the set of all $H$-vertices;
each vertex in the set $H$ is an $H$-vertex.

\subsubsection{Greedy Approximation Algorithm for MHV}
\label{subsec - greedy approxalg for MHV}
{\bf Algorithm {\sc Greedy-MHV}.}
The approximation algorithm {\sc Greedy-MHV} for MHV is quiet simple.
We just color all uncolored vertices by the same color.
Since there are $k$ colors in $C$, we can obtain $k$ vertex colorings for
graph $G$. Finally we output the coloring that has the most number of
happy vertices.

\begin{theorem}
\label{th - 1/k-approximation for MHV}
Algorithm {\sc Greedy-MHV} is a $1/k$-approximation algorithm for the MHV
problem, where $k$ is the number of colors given in the input.
\end{theorem}
\begin{proof}
Let the partial function $c$ be the vertex coloring used in Definition
\ref{def - types of vertices in MHV}.
We partition $L$-vertices further into two subsets $L_P$ and $L_U$.
$L_P$ is the set of uncolored vertices that can become happy (i.e., whose
neighbors have at most one color). $L_U$ is the set of uncolored vertices
that are destined to be unhappy (i.e., whose neighbors already have
at least two distinct colors). Then $(H, P, U, L_P, L_U)$ is a partition of
$V(G)$. Obviously, in the best case $OPT$ can make all vertices in $S$, $P$
and $L_P$ happy, implying $|H| + |P| + |L_P| \geq OPT$.

Let $SOL_i$ be the number of happy vertices when Algorithm {\sc Greedy-MHV}
colors all uncolored vertices by color $i$.
Then we have $|H| + |P| + |L_P| \leq \sum_i SOL_i$.
By the greedy strategy, $SOL$, which is the number of happy vertices found
by {\sc Greedy-MHV}, is at least $\frac1k (|H| + |P| + |L_P|)$.
The theorem follows by observing that {\sc Greedy-MHV} obviously runs in
polynomial time.
\qed
\end{proof}

%

\subsubsection{Subset-Growth Approximation Algorithm for MHV}
\label{subsec - subset-growth approxalg for MHV}
The subset-growth algorithm starts with the partial vertex coloring
$(V_1, V_2, \cdots, V_k)$ defined by the partial function $c$.
From a high level point of view, the algorithm iteratively augments
the subsets in $(V_1, V_2, \cdots, V_k)$ by satisfying the vertices that
can become happy easily at the current time, until $(V_1, V_2, \cdots, V_k)$
becomes a partition of $V(G)$ and thus a vertex coloring is obtained.
This strategy is based on the following further classification of
$L$-vertices, according to the type of their neighbors.
Recall that by Definition \ref{def - types of vertices in MHV},
$L$-vertex means uncolored vertex.

\begin{definition}[Subtypes of $L$-vertex in MHV]
Let $v$ be an $L$-vertex in a vertex coloring. Then,
\begin{enumerate}
\item $v$ is an {\em $L_p$-vertex} if $v$ is adjacent to a $P$-vertex;
\item $v$ is an {\em $L_{h}$-vertex} if
\begin{enumerate}
    \item $v$ is not adjacent to any $P$-vertex,
    \item $v$ can become happy, that is, $v$ is adjacent to $U$-vertices with only
    one color;
\end{enumerate}
\item $v$ is an {\em $L_{u}$-vertex} if
\begin{enumerate}
    \item $v$ is not adjacent to any $P$-vertex,
    \item $v$ is destined to be unhappy, that is, $v$ is adjacent to $U$-vertices
    with more than one colors;
\end{enumerate}
\item $v$ is an {\em $L_f$-vertex} if $v$ is not adjacent to any colored vertex.
\end{enumerate}
\end{definition}

See Figures \ref{fig - process a P-vertex}, \ref{fig - process a Lh-vertex},
\ref{fig - process a Lu-vertex} for examples of the subtypes of $L$-vertex.

The subset-growth algorithm {\sc Growth-MHV} is as follows.

\setcounter{algleo}{0}
\begin{algleo}
\linonumber {\bf Algorithm} {\sc Growth-MHV}
\linonumber {\em Input:} A connected undirected graph $G$ and a partial
    coloring function $c$.
\linonumber {\em Output:} A total vertex coloring for $G$.

\li $\forall 1 \leq i \leq k$, $V_i \leftarrow \{v \colon c(v) = i\}$.
\li \label{step - Growth-MHV - beginning of the main loop}
    {\bf while} there exist $L$-vertices {\bf do}
\begin{algleo}
    \li \label{step - Growth-MHV - process an P-vertex}
        {\bf if} there exists a $P$-vertex $v$ {\bf then}
        \begin{algleo}
            \li $i \leftarrow c(v)$.
            \li Add all the $L_p$-neighbors of $v$ to $V_i$.
                The types of all affected vertices (including $v$
                and vertices in $N^2(v)$) are changed accordingly.
        \end{algleo}
    \li \label{step - Growth-MHV - process an Lh-vertex}
        {\bf elseif} there exists an $L_{h}$-vertex $v$ {\bf then}
        \begin{algleo}
            \li Let $u$ be any $U$-vertex adjacent to $v$,
                then $i \leftarrow c(u)$.
            \li Add $v$ and all its $L$-neighbors to $V_i$.
                The types of all affected vertices
                (including $v$ and vertices in $N^2(v)$) are changed accordingly.
        \end{algleo}
    \li \label{step - Growth-MHV - process an Lu-vertex}
        {\bf else}
        \begin{algleo}
            \linonumber {\em Comment:} There must be an $L_{u}$-vertex.
            \li Let $v$ be any $L_{u}$-vertex, $u$ be the any $U$-vertex
                adjacent to $v$, then $i \leftarrow c(u)$.
            \li Add $v$ to $V_i$.
                The types of all affected vertices
                (including $v$ and vertices in $N(v)$) are changed accordingly.
        \end{algleo}
    \li {\bf endif}
\end{algleo}
\li {\bf endwhile}
\li {\bf return} the vertex coloring $(V_1, V_2, \cdots, V_k)$.
\end{algleo}

When there are still $L$-vertices (i.e., uncolored vertices), Algorithm
{\sc Growth-MHV} works in the following way. It first colors a $P$-vertex's
neighbors to make this $P$-vertex happy
(see Figure \ref{fig - process a P-vertex}).
When there is no any $P$-vertex, it colors an $L_{h}$-vertex and its neighbors
to make the $L_{h}$-vertex happy (see Figure \ref{fig - process a Lh-vertex}).
When there is no any $P$-vertex or $L_{h}$-vertex, it colors an $L_{u}$-vertex
by the color of its any $U$-vertex neighbor
(see Figure \ref{fig - process a Lu-vertex}).
Note that coloring a vertex may generate new $P$-vertices, or $L_{h}$-vertices,
or $L_{u}$-vertices.

\begin{figure}
\begin{center}
\includegraphics*[width=0.55\textwidth]{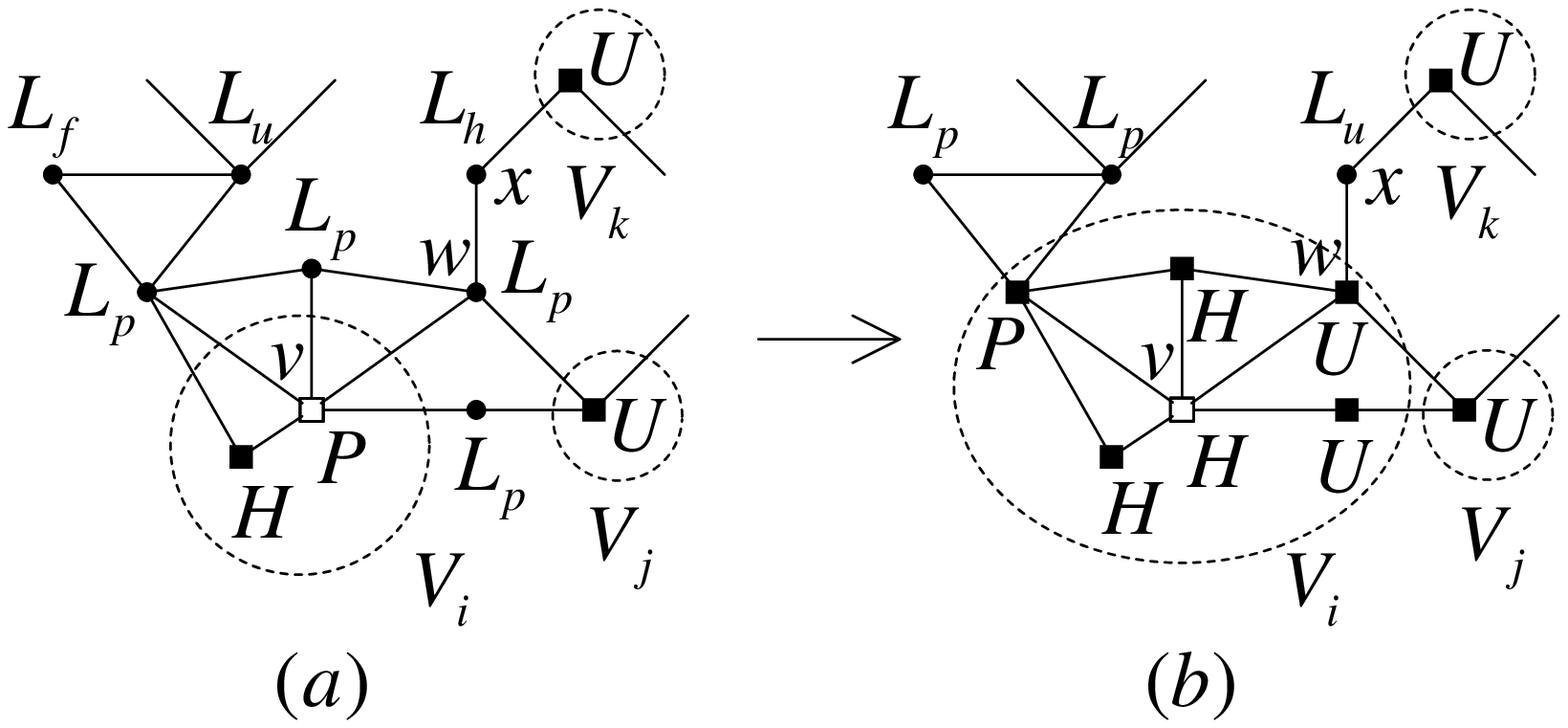}
\end{center}
\caption{Process a $P$-vertex. The hollow vertex $v$ in graph (a) is
the $P$-vertex to be processed. The square vertices mean colored vertices,
while the round vertices mean uncolored vertices.}
\label{fig - process a P-vertex}
\end{figure}

\begin{figure}
\begin{center}
\includegraphics*[width=0.6\textwidth]{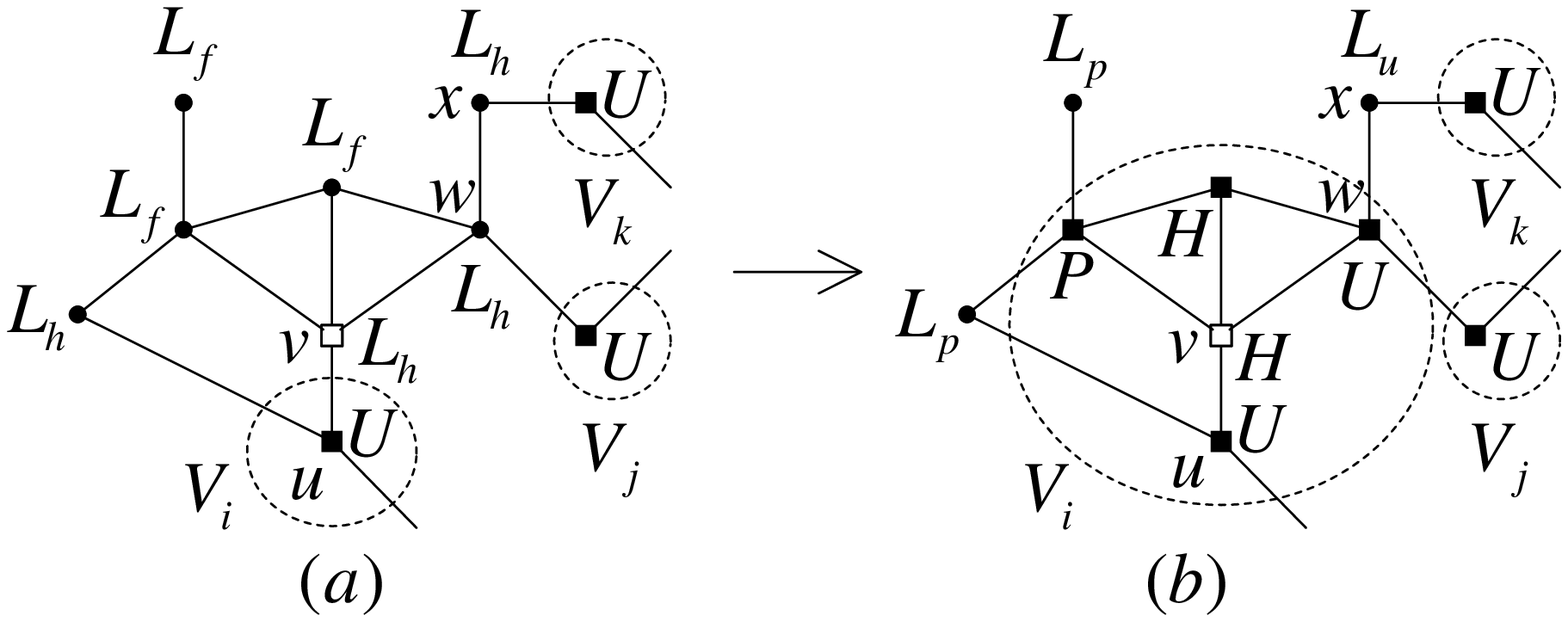}
\end{center}
\caption{Process an $L_{h}$-vertex. The hollow vertex $v$ in graph (a) is
the $L_{h}$-vertex to be processed. Note that when an $L_{h}$-vertex is
to be processed, there is no $P$-vertex in the current graph (a).}
\label{fig - process a Lh-vertex}
\end{figure}

\begin{figure}
\begin{center}
\includegraphics*[width=0.55\textwidth]{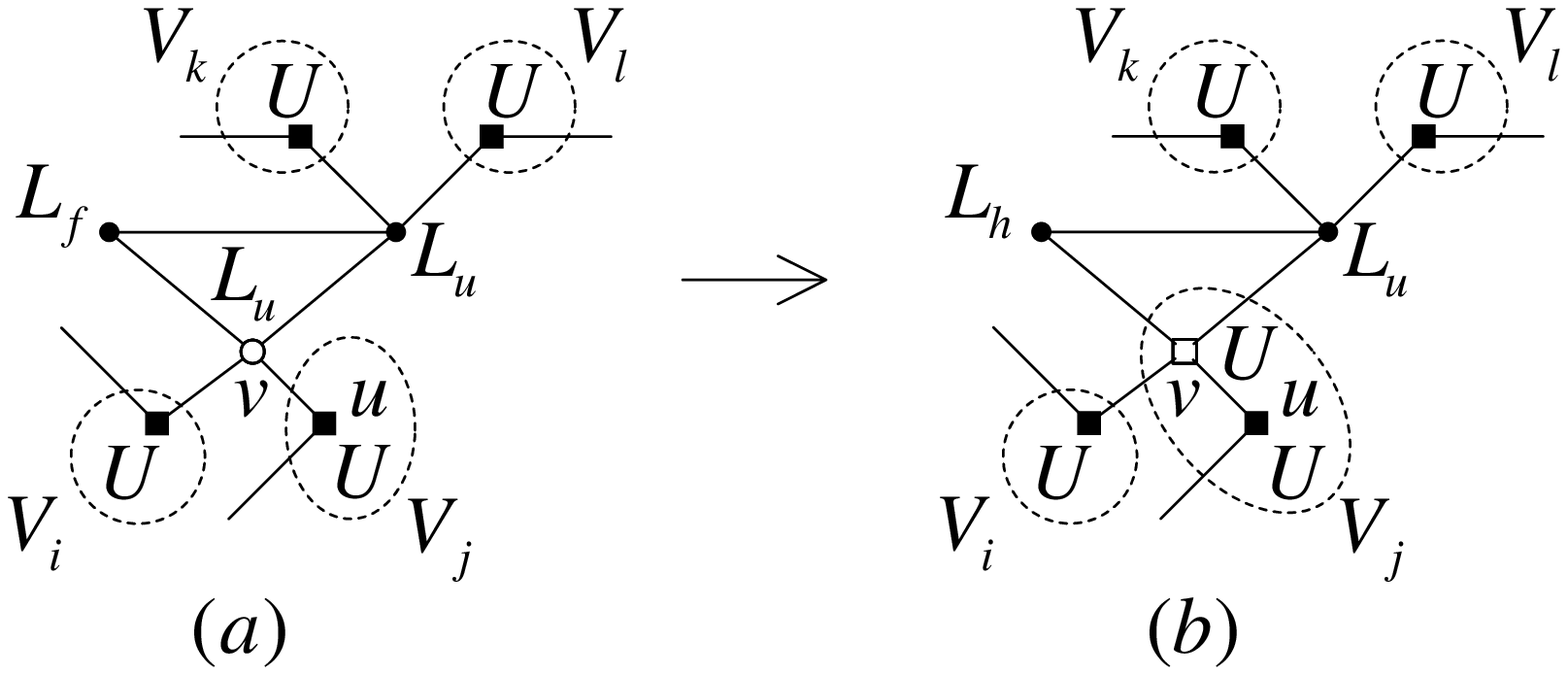}
\end{center}
\caption{Process an $L_{u}$-vertex. The hollow vertex $v$ in graph (a) is
the $L_{u}$-vertex to be processed. Note that when an $L_{u}$-vertex is
to be processed, there is no any $P$-vertex or $L_{h}$-vertex in the current
graph (a).}
\label{fig - process a Lu-vertex}
\end{figure}

When there exist $L$-vertices, it is impossible that there are only
$L_f$-vertices but no any $L_p$-vertex, $L_{h}$-vertex or $L_{u}$-vertex,
since by assumption $G$ is a connected graph and by definition
$L_f$-vertex is not adjacent to any colored vertex.
So, when there isn't any $L_p$-vertex or $L_{h}$-vertex, there must be
at least one $L_{u}$-vertex.
As a result, in step \ref{step - Growth-MHV - process an Lu-vertex}
we don't need an {\bf if} statement like that in
steps \ref{step - Growth-MHV - process an P-vertex} and
\ref{step - Growth-MHV - process an Lh-vertex}.

We use a type name with the superscript ``org'' (means ``original'') to
denote the set of vertices of that type which is determined by
the partial function $c$, and a type name with the superscript ``new'' to
denote the set of vertices of that type which is determined in the execution
of Algorithm {\sc Growth-MHV}. For example, $H^{org}$ is the set of $H$-vertices
that are determined by the partial function $c$, and $H^{new}$ is the set of
$H$-vertices that are newly generated by Algorithm {\sc Growth-MHV}.

Let $\Delta$ be the maximum degree of vertices in the input graph.
We first bound the number of $L_{u}^{new}$-vertices.

\begin{lemma}
\label{lm - |L_u^new| <= Delta (Delta-2) |H^new|}
$|L_{u}^{new}| \leq \Delta (\Delta-2) |H^{new}|$.
\end{lemma}
\begin{proof}
Algorithm {\sc Growth-MHV} iteratively processes three types of vertices,
that is, the $P$-vertices, the $L_{h}$-vertices and the $L_{u}$-vertices.
We will prove the lemma by proving the following three points:
(1) When Algorithm {\sc Growth-MHV} processes a $P$-vertex, at most
$\Delta (\Delta-2)$ $L_{u}^{new}$-vertices are generated,
(2) When Algorithm {\sc Growth-MHV} processes an $L_{h}$-vertex, at most
$(\Delta-1)(\Delta-2)$ $L_{u}^{new}$-vertices are generated, and
(3) When Algorithm {\sc Growth-MHV} processes an $L_{u}$-vertex,
no $L_{u}^{new}$-vertex is generated.

Consider the first point. Let $v$ be a $P$-vertex to be processed.
Suppose $v$ has an $L_p$-neighbor $w$, which is adjacent to a $U$-vertex.
Only if there is an $L_{h}$-vertex $x$ which is the neighbor of $w$,
$x$ will become a newly generated $L_{u}$-vertex when the $P$-vertex $v$
is processed. See Figure \ref{fig - process a P-vertex} for an example.
Since the maximum vertex degree is $\Delta$, $v$ has at most $\Delta$
$L_p$-neighbors, and $w$ has at most $\Delta - 2$ $L_{h}$-neighbors.
This implies that when $v$ is processed, at most $\Delta (\Delta-2)$
$L_{u}^{new}$-vertices can be generated.

Then consider the second point. Suppose the $L_{h}$-vertex to be processed
is $v$. Suppose $v$ has an $L$-neighbor $w$ ($w$ can be an $L_{h}$-vertex or
an $L_{u}$-vertex), which is adjacent to a $U$-vertex. Similarly,
only if there is an $L_{h}$-vertex $x$ which is the neighbor of $w$,
$x$ will become a newly generated $L_{u}$-vertex when the $L_{h}$-vertex $v$
is processed. See Figure \ref{fig - process a Lh-vertex} for an example.
Since the maximum vertex degree is $\Delta$, $v$ has at most $\Delta - 1$
$L$-neighbors, and $w$ has at most $\Delta - 2$ $L_{h}$-neighbors.
This implies that when $v$ is processed, at most $(\Delta-1)(\Delta-2)$
$L_{u}^{new}$-vertices can be generated.

Finally consider the third point. When Algorithm {\sc Growth-MHV} processes
an $L_{u}$-vertex, there is no any $L_{h}$-vertex (or $P$-vertex) in the
current graph. So, adding an $L_{u}$-vertex to some subset $V_i$
does not generate any new $L_{u}$-vertex.
See Figure \ref{fig - process a Lu-vertex} for an example.

When Algorithm {\sc Growth-MHV} processes a $P$-vertex or an $L_{h}$-vertex,
at least one vertex becomes an $H$-vertex. So we can charge the number of
newly generated $L_{u}$-vertices to this newly generated $H$-vertex.
This finishes the proof of the lemma.
\qed
\end{proof}

The following Lemma \ref{lm - upper bound on OPT} gives an upper bound
on $OPT$, the number of happy vertices in an optimal solution to
the $k$-MHV problem.

\begin{lemma}
\label{lm - upper bound on OPT}
$OPT \leq |H^{org}| + (\Delta + 1)(|L^{org}| - |L_{u}^{org}|)$.
\end{lemma}
\begin{proof}
By the partial function $c$, all vertices in the original graph
(i.e., the input graph that has not been colored by Algorithm {\sc Growth-MHV})
are partitioned into four vertex subsets $H^{org}$, $P^{org}$, $U^{org}$
and $L^{org}$. Subset $L^{org}$ is further partitioned into four subsets
$L_p^{org}$, $L_{h}^{org}$, $L_{u}^{org}$ and $L_f^{org}$.
By definition, all vertices in $U^{org}$ are unhappy.
And, all vertices in $L_{u}^{org}$ are destined to be unhappy
since each of them is adjacent to at least two vertices with different
colors. So, in the best case all vertices in $P^{org}$ and $L^{org}$
except those in $L_{u}^{org}$ would be happy.
Noticing that the vertices in $H^{org}$ are already happy, we have
\begin{equation}
OPT \leq |H^{org}| + |P^{org}| + |L^{org}| - |L_{u}^{org}|. \nonumber
\end{equation}

Since each $P$-vertex must be adjacent to some $L_p$-vertex, and each
$L_p$-vertex can be adjacent to at most $\Delta$ $P$-vertices, the number
of $P^{org}$-vertices is at most $\Delta |L_p^{org}|$.
Since $|L_p^{org}| \leq |L^{org}| - |L_{u}^{org}|$, we get that
\begin{eqnarray}
OPT &\leq& |H^{org}| + \Delta |L_p^{org}| + |L^{org}| - |L_{u}^{org}|
\nonumber \\
&\leq& |H^{org}| + (\Delta + 1)(|L^{org}| - |L_{u}^{org}|),
\nonumber
\end{eqnarray}
concluding the lemma.
\qed
\end{proof}

\begin{lemma}
\label{lm - lower bound on |H^new|}
$|H^{new}| \geq \frac{1}{\Delta(\Delta-1)}(|L^{org}| - |L_{u}^{org}|)$.
\end{lemma}
\begin{proof}
Recall that there are four subtypes of an $L$-vertex, i.e., $L_p$-vertex,
$L_{h}$-vertex, $L_{u}$-vertex and $L_f$-vertex.
Among them only $L_p$-vertex and $L_{h}$-vertex will (directly) contribute
to generating $H$-vertices.
For an $L_f$-vertex, it will ultimately become one of the other three
types of $L$-vertex.
For each $L_{u}$-vertex, although it may become an $L_p$-vertex and hence
can contribute to generating $H$-vertices, in the worst case we may
assume that it is added to some subset $V_i$
and contribute nothing to the generation of $H$-vertex.

By step \ref{step - Growth-MHV - process an P-vertex} and
step \ref{step - Growth-MHV - process an Lu-vertex},
each time an $H$-vertex is generated, at most $\Delta$ $L_p$-vertices or
$L_{h}$-vertices are consumed (i.e., colored). Furthermore, once
an $L$-vertex is colored, it will never be re-colored or de-colored.
So we have
\begin{equation}
|H^{new}| \geq \frac{1}{\Delta}(|L^{org}| - |L_{u}^{org}| - |L_{u}^{new}|).
\nonumber
\end{equation}

By Lemma \ref{lm - |L_u^new| <= Delta (Delta-2) |H^new|},
we have
\begin{eqnarray}
\frac{1}{\Delta}(|L^{org}| - |L_{u}^{org}| - |L_{u}^{new}|)
&\geq& \frac{1}{\Delta}(|L^{org}| - |L_{u}^{org}| - \Delta(\Delta-2)|H^{new}|)
\nonumber \\
&=& \frac{1}{\Delta}(|L^{org}| - |L_{u}^{org}|) - (\Delta-2)|H^{new}|.
\nonumber
\end{eqnarray}
Therefore, $(\Delta-1) |H^{new}| \geq \frac{1}{\Delta}(|L^{org}| - |L_{u}^{org}|)$.
The lemma follows.
\qed
\end{proof}

\begin{theorem}
The MHV problem can be approximated within a factor of $\Omega(\Delta^{-3})$
in polynomial time.
\end{theorem}
\begin{proof}
Algorithm {\sc Growth-MHV} obviously runs in polynomial time.
Let $SOL$ be the number of happy vertices found by Algorithm {\sc Growth-MHV}.
Then we have
\begin{eqnarray}
SOL &=& |H^{org}| + |H^{new}| \nonumber \\
&\geq& |H^{org}| + \frac{1}{\Delta(\Delta-1)}\Bigl(|L^{org}| - |L_{u}^{org}|\Bigr)
\quad \mbox{(By Lemma \ref{lm - lower bound on |H^new|})}
\nonumber \\
&\geq& \frac{1}{\Delta(\Delta-1)(\Delta+1)}\Bigl(|H^{org}| + (\Delta+1)(|L^{org}| - |L_{u}^{org}|)\Bigr)
\nonumber \\
&\geq& \frac{1}{\Delta(\Delta-1)(\Delta+1)} OPT
\qquad \qquad \qquad \mbox{(By Lemma \ref{lm - upper bound on OPT})}
\nonumber \\
&=& \Omega(\Delta^{-3}) OPT. \nonumber
\end{eqnarray}
The theorem follows.
\qed
\end{proof}


\section{Algorithms for MHE}
\label{sec - algorithms for MHE}

\subsection{2-MHE Is in P}
For 2-MHE, the partial function $c$ can only use two colors, to say,
color 1 and color 2. Given such an instance, merge all vertices with color 1
assigned by $c$ into a single vertex $s$, and all vertices with color 2
into a single vertex $t$. (The edges whose two endpoints are merged disappear
in the procedure.) Then compute a minimum $s$-$t$ cut $(V_1, V_2)$ on
the resulting instance. Suppose $s \in V_1$ and $t \in V_2$. Assign color 1
to all vertices (including the merged vertices) in $V_1$, and color 2 to
all vertices in $V_2$. Since $(V_1, V_2)$ is a minimum $s$-$t$ cut,
the number of happy edges in the resulting vertex coloring is maximized.
By the work of \cite{ET75}, a maximum flow (and hence a minimum $s$-$t$)
in a unit capacity network can be computed in $O(\min\{n^{2/3}m, m^{3/2}\})$
time. So we have

\begin{theorem}
\label{th - 2-MHE is in P}
The 2-MHE problem can be solved in $O(\min\{n^{2/3}m, m^{3/2}\})$ time.
\qed
\end{theorem}

\subsection{Approximation Algorithm for MHE}
The MHE problem admits a simple division-strategy based algorithm which
yields a $1/2$-approximation. The algorithm is designed to
deal with more general graphs with nonnegative weights $\{w(e)\}$ defined
on edges. We thus denote by $w(E')$ the total weight of edges in an edge
subset $E'$.

\setcounter{algleo}{0}
\begin{algleo}
\linonumber {\bf Algorithm} {\sc Division-MHE}
\linonumber {\em Input:} An undirected graph $G$ and a partial coloring
    function $c$.
\linonumber {\em Output:} A total vertex coloring for $G$.
\li $G_1 \leftarrow G$.
\li Let $E'$ be the set of edges in $G_1$ that has exactly one endpoint not
    colored by function $c$. Define graph $G' = (V(G_1), E')$,
    which is a subgraph of $G_1$.
\li \label{step - Division-MHE - Graph G' and its stars}
    For each star $S$ in $G'$ centered at an uncolored vertex $v$, color $v$
    by a color in $\{c(u) \mid u \in N(v), u \mbox{ is colored} \}$
    such that the total weight of happy edges in $S$ is maximized.
\li Color all vertices in $G_1$ still having not been colored by just one
    arbitrary color. Denote by $SOL_1$ the vertex coloring of $G_1$.
\li $G_2 \leftarrow G$.
\li Color all uncolored vertices in $G_2$ by just one arbitrary color.
    Denote by $SOL_2$ the vertex coloring of $G_2$.
\li {\bf return} the better one among $SOL_1$ and $SOL_2$.
\end{algleo}

Algorithm {\sc Division-MHE} computes two independent solutions
$SOL_1$ and $SOL_2$ to graph $G$, and then outputs the better one,
where the better one means the solution making more edges happy.
For an illustration of graph $G'$ and its stars in
step \ref{step - Division-MHE - Graph G' and its stars},
please refer to Figure \ref{fig - stars in graph G'}.

\begin{figure}
\begin{center}
\includegraphics*[width=0.6\textwidth]{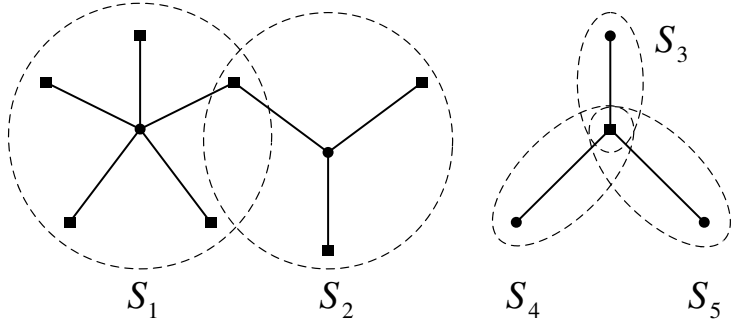}
\end{center}
\caption{An example of graph $G'$. Each edge in $G'$ has its one endpoint
colored and the other endpoint uncolored. The square vertices mean colored
vertices, while the round vertices mean uncolored vertices. Each star
(marked with dashed circle) is centered at an uncolored vertex.
Two stars (e.g., $S_1$ and $S_2$) may share common colored vertices.}
\label{fig - stars in graph G'}
\end{figure}

\begin{theorem}
\label{th - algorithm Division-MHE is a 1/2-approxalg}
Algorithm {\sc Division-MHE} is a $1/2$-approximation algorithm for
the MHE problem.
\end{theorem}
\begin{proof}
First, the algorithm obviously runs in polynomial time.

Let $W^{org}$ be the total weight of edges already being happy by
the partial coloring function $c$. This weight can be trivially obtained
by any solution.

Let $W'$ be the total weight of happy edges found by Algorithm
{\sc Division-MHE} on graph $G'$. Note that $W'$ is the maximum
total weight that can be obtained from graph $G'$.
Let $E''$ be the set of edges that has both of its two endpoints uncolored
by function $c$, and $W'' = w(E'')$ be its total weight.
Then we have $OPT \leq W^{org} + W' + W''$.

By the algorithm, we know $SOL_1 \geq W^{org} + W'$ and
$SOL_2 \geq W^{org} + W''$. Then the approximation ratio $1/2$ of
{\sc Division-MHE} is obvious
since $SOL = \max\{SOL_1, SOL_2\} \geq \frac12 (W^{org} + W' + W'')$.
\qed
\end{proof}


\section{Hardness Results}
\label{sec - Hardness results}

\subsection{NP-hardness of MHE}
\label{subsec - NP-hardness of MHE}
The NP-hardness of the $3$-MHE problem is proved by a reduction from
the Multiway Cut problem \cite{DJP+94}.

\begin{theorem}
\label{th - 3-MHE is NP-hard}
The 3-MHE problem is NP-hard.
\end{theorem}
\begin{proof}
Given an undirected graph $G = (V, E)$ and a terminal set $D = \{s_1, s_2, s_3\}$,
the 3-Terminal Cut problem (i.e., the Multiway Cut problem with
3 terminals), which is NP-hard \cite{DJP+94}, asks for a minimum cardinality
edge set such that its removal from $G$ disconnects the three terminals
from one another. Given an instance $(G, D)$ of 3-Terminal Cut, we construct
the instance $(H, C, c)$ of 3-MHE as follows.
Graph $H$ is just $G$. Color set $C$ is set to be $\{1, 2, 3\}$.
The partial function $c$ assigns colors 1, 2, 3 to vertices $s_1, s_2, s_3$,
respectively. Let $c^*$ be the cardinality of an optimal 3-way cut for
$(G, D)$, and $m^*$ be the number of happy edges of an optimal vertex
coloring for $(H, C, c)$. Then one can easily find that $m^* = m - c^*$,
where $m = |E(G)|$ ($=|E(H)|$). This shows the 3-MHE problem is NP-hard.
\qed
\end{proof}

\begin{corollary}
The MHE problem is NP-hard.
\end{corollary}
\begin{proof}
In the input of MHE, just set $k$ to be 3. \qed
\end{proof}


\begin{theorem}
\label{th - k-MHE is NP-hard}
The $k$-MHE problem is NP-hard for any constant $k \geq 3$.
\end{theorem}
\begin{proof}
By Theorem \ref{th - 3-MHE is NP-hard}, we need only focus on $k > 3$.
Let $k$ be such a constant.

Given a 3-MHE instance $(G, c)$, we construct a $k$-MHE instance $(G', c')$
as follows. Build $2(k-3)$ vertices $x_4, y_4, x_5, y_5, \cdots, x_k, y_k$
and $k-3$ edges $(x_i, y_i)$, $4 \leq i \leq k$.
Vertices $x_i$ and $y_i$ are colored by color $i$, for $4 \leq i \leq k$.
Let $v$ be a vertex in $G$ whose color given by $c$ is 1.
Then put $k-3$ edges $(v, x_i)$, $4 \leq i \leq k$.
This is our new instance $(G', c')$.

Obviously for $4 \leq i \leq k$, each edge $(x_i, y_i)$ is happy
whereas each edge $(v, x_i)$ is unhappy.
So, the optimum of $(G, c)$ is just equal to the optimum of $(G', c')$
minus $k-3$, concluding the theorem.
\qed
\end{proof}

\subsection{NP-hardness of MHV}
\label{subsec - NP-hardness of MHV}
\begin{theorem}
The $k$-MHV problem is NP-hard for any constant $k \geq 3$.
\end{theorem}
\begin{proof}
By Theorem \ref{th - k-MHE is NP-hard}, $k$-MHE is NP-hard ($k \geq 3$).
We thus reduce $k$-MHE to $k$-MHV.

Let $(G, c)$ be a $k$-MHE instance. The instance $(G', c')$ of $k$-MHV
is constructed as follows.
Add $k$ vertices $x_1, x_2, \cdots, x_k$ and put an edge between $x_i$
and $v$, for each $1 \leq i \leq k$ and each $v \in V(G)$.
Vertex $x_i$ is colored by $i$, for $1 \leq i \leq k$.
For every edge $(u, v) \in E(G)$, add a vertex $y_{uv}$ and replace the edge
by two edges $(u, y_{uv})$ and $(y_{uv}, v)$.
This is our new instance $(G', c')$.

Since in graph $G$ there are vertices with pre-specified colors,
each $x_i$ ($1 \leq i \leq k$) cannot become happy no matter how the remaining
vertices are colored. Every original vertex $v \in V(G)$ also cannot become happy
since it is adjacent to all $x_i$'s.
Let $(u, v)$ be any edge in $G$. Since the degree of vertex $y_{uv}$ is 2,
it is happy iff its two neighbors have the same color.
This shows that the optimum of the $k$-MHE instance $(G, c)$ is equal to
the optimum of the $k$-MHV instance $(G', c')$.
The theorem follows.
\qed
\end{proof}

\section{Conclusions}
\label{sec - conclusions}
The MHV problem and the MHE problem are two natural graph coloring problems
arising in the homophyly phenomenon of networks. In this paper we prove
the NP-hardness of the MHV problem and the MHE problem, and give several
approximation algorithms for these two problems.

Since our algorithms {\sc Greedy-MHV}, {\sc Growth-MHV} and {\sc Division-MHE}
actually do not care whether the color number $k$ is given in the input or
whether $k$ is a constant, the $k$-MHV and $k$-MHE problems can also be
approximated within $\max\{1/k, \Omega(\Delta^{-3})\}$ and $1/2$, respectively.

To improve the approximation ratios for MHV and MHE remains an immediate
open problem. It is also interesting to study the MHV and MHE problems
in random graphs generated from the classical network models, and in the
real-world large networks.

\section*{Acknowledgements}
We are grateful for fruitful discussions on this paper with
Dr. Mingji Xia at Institute of Software, Chinese Academy of Sciences.

\appendix

\section*{Appendix}

\section{Variants of MHV}
\label{sec - variants of MHV}
For a vertex $v$ in the MHV problem, instead of requiring that {\em all}
neighbors of $v$ have the same color as that of $v$, to make $v$ happy
we may only require at least $\rho \cdot \deg(v)$ neighbors have the same
color as that of $v$, or only require at least $q$ neighbors have the color
identical to that of $v$, for some global number $q$. This leads to two
natural variants of the MHV problem, that is, the SoftMHV problem and
the HardMHV problem. Similarly, we can define the corresponding varints
for the $k$-MHV problem, and our results in this section naturally extends
to these variants. For simplicity, we only consider approximation algorithms
for the SoftMHV and HardMHV problems.

Fix a vertex coloring, and let $v$ be a (colored or uncolored) vertex.
Define $N_i(v)$ to be the set of vertices in $N(v)$ which has color $i$,
for $1 \leq i \leq k$.

\section{MHV with Soft Threshold}
Let $\rho$ be a number in $(0, 1)$. In the soft-threshold extension of
the MHV problem (SoftMHV for short), a vertex $v$ is happy if $v$ is
colored and $|N^s(v)| \geq \rho \cdot \deg(v)$.
Given a connected undirected graph $G$, a partial coloring function $c$,
the SoftMHV problem asks for a total vertex coloring extended from $c$ that
maximizes the number of happy vertices.
(The number $\rho$ can be given as a part of the input or be a constant.
We do not distinguish between these two cases for simplicy.)

\subsection{Algorithm for SoftMHV}
As what is done in Definition \ref{def - types of vertices in MHV},
we define the types of vertices according to the given vertex coloring.

\begin{definition}[Types of vertex in SoftMHV]
\label{def - types of vertex - SoftMHV}
Fix a (partial or total) vertex coloring. Let $v$ be a vertex. Then,
\begin{enumerate}
\item $v$ is an $H$-vertex if $v$ is colored and happy;

\item $v$ is a $U$-vertex if
\begin{enumerate}
\item $v$ is colored, and
\item $v$ is destined to be unhappy, (i.e., $\deg(v) - |N^d(v)| < \rho \cdot \deg(v)$);
\end{enumerate}

\item $v$ is a $P$-vertex if
\begin{enumerate}
\item $v$ is colored,
\item $v$ has not been happy (i.e., $|N^s(v)| < \rho \cdot \deg(v)$), and
\item $v$ can become an $H$-vertex (i.e., $|N^s(v)| + |N^u(v)| \geq \rho \cdot \deg(v)$);
\end{enumerate}

\item $v$ is an $L$-vertex if $v$ has not been colored.
\end{enumerate}
\end{definition}

We note that Algorithm {\sc Greedy-MHV} is also a $1/k$-approximation
algorithm for the SoftMHV problem. To see this, we just define $L_P$
in Theorem \ref{th - 1/k-approximation for MHV} as the set of uncolored
vertices $v$ such that $|N^u(v)| + \max \{|N_i(v)|\} \geq \rho \cdot \deg(v)$,
and $L_D = L - L_P$.

\begin{theorem}
The SoftMHV problem can be approximated within a factor of $1/k$
in polynomial time. \qed
\end{theorem}

Below we give the subset-growth approximation algorithm {\sc Growth-SoftMHV}
for the SoftMHV problem. First we define the subtypes of $L$-vertex.

\begin{definition}[Subtypes of $L$-vertex in SoftMHV]
Let vertex $v$ be an $L$-vertex in a vertex coloring. Then,
\begin{enumerate}
\item $v$ is an $L_p$-vertex if $v$ is adjacent to a $P$-vertex,

\item $v$ is an $L_{h}$-vertex if
\begin{enumerate}
\item $v$ is not adjacent to any $P$-vertex,
\item $v$ is adjacent to an $H$-vertex or a $U$-vertex, and
\item $v$ can become happy (that is,
$|N^u(v)| + \max \{|N_i(v)| \colon 1 \leq i \leq k \} \geq \rho \cdot \deg(v)$),
\end{enumerate}

\item $v$ is an $L_{u}$-vertex if
\begin{enumerate}
\item $v$ is not adjacent to any $P$-vertex,
\item $v$ is adjacent to an $H$-vertex or a $U$-vertex, and
\item $v$ is destined to be unhappy (that is,
$|N^u(v)| + \max \{|N_i(v)| \colon 1 \leq i \leq k \} < \rho \cdot \deg(v)$),
\end{enumerate}

\item $v$ is an $L_f$-vertex if $v$ is not adjacent to any colored
vertex.
\end{enumerate}
\end{definition}

\setcounter{algleo}{0}
\begin{algleo}
\linonumber {\bf Algorithm} {\sc Growth-SoftMHV}
\linonumber {\em Input:} A connected undirected graph $G$ and a partial
    coloring function $c$.
\linonumber {\em Output:} A total vertex coloring for $G$.

\li $\forall 1 \leq i \leq k$, $V_i \leftarrow \{v \colon c(v) = i\}$.
\li \label{step - Growth-SoftMHV - beginning of the main loop}
    {\bf while} there exist $L$-vertices {\bf do}
\begin{algleo}
    \li {\bf if} there exists a $P$-vertex $v$ {\bf then}
        \begin{algleo}
            \li $i \leftarrow c(v)$.
            \li \label{step - Growth-SoftMHV - process P-vertex}
                Add its any $\lceil \rho \cdot \deg(v) \rceil - |N^s(v) \cap V_i|$
                $L_p$-neighbors to vertex subset $V_i$.
                The types of all affected vertices (including $v$ and vertices
                in $N^2(v)$) are changed accordingly.
        \end{algleo}
    \li {\bf elseif} there exists an $L_{h}$-vertex $v$ {\bf then}
        \begin{algleo}
            \li Let $V_i$ be the vertex subset in which $v$ has
                the maximum colored neighbors.
            \li \label{step - Growth-SoftMHV - process Lh-vertex}
                Add vertex $v$ and its any $\lceil \rho \cdot \deg(v) \rceil - |N^s(v) \cap V_i|$
                $L$-neighbors to vertex subset $V_i$.
                The types of all affected vertices (including $v$ and vertices
                in $N^2(v)$) are changed accordingly.
        \end{algleo}
    \li {\bf else}
        \begin{algleo}
            \linonumber {\em Comment:} There must be an $L_{u}$-vertex.
            \li Let $v$ be any $L_{u}$-vertex, and $V_i$ be any vertex subset in which
                $v$ has colored neighbors.
            \li Add vertex $v$ to subset $V_i$.
                The types of all affected vertices (including $v$ and vertices
                in $N(v)$) are changed accordingly.
        \end{algleo}
    \li {\bf endif}
\end{algleo}
\li {\bf endwhile}
\li {\bf return} the vertex coloring $(V_1, V_2, \cdots, V_k)$.
\end{algleo}

In step \ref{step - Growth-SoftMHV - process P-vertex},
the algorithm adds the least number
(that is, $\lceil \rho \cdot \deg(v) \rceil - |N^s(v) \cap V_i|$)
of $v$'s neighbors to subset $V_i$ to make $v$ happy. The same thing is done
in step \ref{step - Growth-SoftMHV - process Lh-vertex}.

\begin{lemma}
\label{lm - lower bound on |H^new|, Growth-SoftMHV}
$|L_{u}^{new}| \leq O(\Delta^2) |H^{new}|$.
\end{lemma}
\begin{proof}
Suppose Algorithm {\sc Growth-SoftMHV} is to process a $P$-vertex $v$,
which is already colored by color $i$.
When $v$ is processed, at most $\lceil \rho \Delta \rceil$ $L_p$-neighbors
of $v$ are added to $V_i$. Each of the $L_p$-neighbors has at most
$\Delta - 1$ $L_{h}$-neighbors.
In the worst case, all these $L_{h}$-neighbors,
plus the remaining $L_p$-neighbors of $v$, could become $L_{u}$-vertices
when $v$ is processed. So, at most
$\lceil \rho \Delta \rceil (\Delta-1) + (1-\alpha) \Delta = O(\Delta^2)$
$L_{u}^{new}$-vertices can be generated in this case.

Then suppose the algorithm is to process an $L_{h}$-vertex $v$. Let
$V_i$ be the vertex subset in which $v$ has the maximum colored neighbors.
When $v$ is processed, at most $\lceil \rho \Delta \rceil - 1$ $L$-neighbors
of $v$ are added to $V_i$. Each of these $L$-neighbors can have at most
$\Delta - 1$ $L_{h}$-neighbors. In the worst case,
all these $L_{h}$-neighbors, plus the remaining $L$-neighbors of $v$,
could become $L_{u}$-vertices when $v$ is processed.
So, at most
$(\lceil \rho \Delta \rceil - 1) (\Delta-1) + (1-\alpha) \Delta = O(\Delta^2)$
$L_{u}^{new}$-vertices can be generated in this case.

When the algorithm processes an $L_{u}$-vertex, there are only
$L_{u}$-vertices or $L_f$-vertices (if any) in the current graph.
So, coloring an $L_{u}$-vertex does not generate any new $L_{u}$-vertex.

By charging the number of newly generated $L_{u}$-vertices to the newly
generated $H$-vertex, we finish the proof of the lemma.
\qed
\end{proof}

\begin{theorem}
The SoftMHV problem can be approximated within a factor of
$\Omega(\Delta^{-3})$ in polynomial time.
\end{theorem}
\begin{proof}
Each time an $H$-vertex is generated, at most $\lceil \rho \Delta \rceil$
$L$-vertices are consumed (i.e., colored).
So, for the number of newly generated $H$-vertices we have
$|H^{new}| \geq (|L^{org}| - |L_{u}^{org}| - |L_{u}^{new}|) / \lceil \rho \Delta \rceil$.
By Lemma \ref{lm - lower bound on |H^new|, Growth-SoftMHV}, we get
\begin{equation}
|H^{new}| \geq \frac{|L^{org}| - |L_{u}^{org}|}{O(\Delta^2)}.
\nonumber
\end{equation}

Let $OPT$ be the number of happy vertices in an optimal solution to
the problem. By the same reason as in Lemma \ref{lm - upper bound on OPT},
we obtain
\begin{eqnarray}
OPT &\leq& |H^{org}| + |P^{org}| + |L^{org}| - |L_{u}^{org}| \nonumber \\
&\leq& |H^{org}| + \Delta |L_p^{org}| + |L^{org}| - |L_{u}^{org}| \nonumber \\
&\leq& |H^{org}| + (\Delta+1) (|L^{org}| - |L_{u}^{org}|). \nonumber
\end{eqnarray}

Let $SOL$ be the number of happy vertices found by Algorithm
{\sc Growth-SoftMHV}. Then we have
\begin{eqnarray}
SOL &=& |H^{org}| + |H^{new}| \nonumber \\
&\geq& |H^{org}| + \frac{1}{O(\Delta^2)}\Bigl(|L^{org}| - |L_{u}^{org}|\Bigr)
\nonumber \\
&\geq& \frac{1}{O(\Delta^3)}\Bigl(|H^{org}| + \Delta(|L^{org}| - |L_{u}^{org}|)\Bigr)
\nonumber \\
&=& \Omega(\Delta^{-3}) OPT. \nonumber
\end{eqnarray}

Finally, notice that Algorithm {\sc Growth-SoftMHV} obviously runs
in polynomial time. This gives the theorem.
\qed
\end{proof}

\subsection{NP-Hardness of SoftMHV}
\begin{theorem}
For any real number $0 < \rho < 1$, there exist infinitely many integers
$k \geq 3$, such that the corresponding SoftMHV problem is NP-hard.
\end{theorem}
\begin{proof}
Reduce from 3-MHE. Let $(G, c)$ be an instance of 3-MHE, and $\rho$ be any
real constant in $(0, 1)$. We shall construct a SoftMHV instance
$(G', c')$ in the following, in which the color number $k \geq 3$ is an integer
that depends only on $\rho$. The value of $k$ will be given later.

Let $h$ be an integer constant depending on $\rho$ and $k$, which will be
fixed later.
For every edge $(u, v) \in E(G)$, add $h+k+1$ vertices $x_{uv}$ (called
$x$-vertex), $y_1^{uv}$, $y_2^{uv}$, $\cdots$, $y_h^{uv}$ (called $y$-vertices),
$z_1^{uv}$, $z_2^{uv}$, $\cdots$, $z_q^{uv}$ (called $z$-vertices).
Replace edge $(u, v)$ by two consecutive edges $(u, x_{uv})$
and $(x_{uv}, v)$. For each vertex $a \in \{y_1^{uv}, \cdots, y_h^{uv},$
$z_1^{uv}, \cdots, z_k^{uv} \}$, connect it to $x_{uv}$ via an edge
$(a, x_{uv})$.
For $1 \leq i \leq k$, vertex $z_i^{uv}$ has a pre-specified color $i$.

For every vertex $v \in V(G)$, add $\Delta \cdot k$ vertices $w_{1,1}^v$,
$w_{1,2}^v$, $\cdots$, $w_{1,\Delta}^v$, $w_{2,1}^v$, $w_{2,2}^v$, $\cdots$,
$w_{2, \Delta}^v$, $\cdots$, $w_{k,1}^v$, $w_{k,2}^v$, $\cdots$,
$w_{k,\Delta}^v$ (called $w$-vertices), where $\Delta$ is the maximum vertex
degree of $G$.
For each $1 \leq i \leq q$ and each $1 \leq j \leq \Delta$, connect vertex
$w_{i,j}^v$ to $v$ via an edge $(w_{i, j}^v, v)$. Vertex $w_{i,j}^v$ is
colored in advance by color $i$, $1 \leq i \leq k$, $1 \leq j \leq \Delta$.
This is our graph $G'$ in the new instance of SoftMHV.

Next we determine constants $h$ and $k$.
To enable the reduction to work, $h$ and $k$ should satisfy
\begin{eqnarray}
h + 3 &\geq& \rho (h + k + 2), \label{eqn - h + 3 >= rho(h + k + 2)}\\
h + 2 &<& \rho (h + k + 2). \label{eqn - h + 2 < rho (h + k + 2)}
\end{eqnarray}
Let $(u, v)$ be any edge in $G$. Consider vertex $x_{uv}$ in $G'$.
No matter how $x_{uv}$ is colored (recall that the color set is
$\{1, 2, \cdots, k\}$), there is exactly one vertex in
$\{z_1^{uv}, \cdots, z_q^{uv}\}$ having the same color as that of $x_{uv}$.
Note that $\deg_{G'}(x_{uv}) = h + k + 2$.
So, inequality (\ref{eqn - h + 3 >= rho(h + k + 2)}) guarantees that
if all vertices in $\{u, v, y_1^{uv}, \cdots, y_2^{uv}\}$ have the same
color as that of $x_{uv}$, $x_{uv}$ will be happy, and,
inequality (\ref{eqn - h + 2 < rho (h + k + 2)}) guarantees that
if there is one vertex in $\{u, v, y_1^{uv}, \cdots, y_2^{uv}\}$ having
different color to that of $x_{uv}$, $x_{uv}$ will be unhappy.

By inequality (\ref{eqn - h + 3 >= rho(h + k + 2)}) and
inequality (\ref{eqn - h + 2 < rho (h + k + 2)}), the value of integer $h$
should satisfy
\begin{equation}
\label{eqn - interval of h}
h \in \Bigl [ \frac{\rho k + 2\rho - 3}{1 - \rho}, \frac{\rho k + 2\rho - 2}{1 - \rho} \Bigr).
\end{equation}
Since $\rho k + 2\rho - 2 = (\rho k + 2\rho - 3) + 1$ and $1 - \rho < 1$,
there must be at least one integer in the interval of
(\ref{eqn - interval of h}).

Of course, the left end of the interval of (\ref{eqn - interval of h})
should be at least 1. This gives
\begin{equation}
\label{eqn - k >= 4/rho - 3}
k \geq \frac4\rho - 3.
\end{equation}

For each vertex $v \in V(G')$ that comes from $G$, we want to guarantee that
no matter how the vertices in $G'$ are colored, $v$ will never be happy.
Note that no matter what color vertex $v$ is colored by,
there are exactly $\Delta$ vertices in
$\{w_{i,j}^v \colon 1 \leq i \leq k, 1 \leq j \leq \Delta \}$ having the
same color as that of $v$.
Since $\deg_G(v) \leq \Delta$ and $\deg_{G'}(v) \geq k\Delta + 1$,
to make vertex $v$ unsatisfiable, we just need $2\Delta / (k \Delta + 1) < \rho$.
Since $2\Delta / (k \Delta + 1) < 2\Delta / (k \Delta)$, this will be
guaranteed by letting
\begin{equation}
\label{eqn - k >= 2/rho}
k \geq \frac2\rho.
\end{equation}

Since we start our reduction from the 3-MHE problem, naturally we need
\begin{equation}
\label{eqn - k >= 3}
k \geq 3.
\end{equation}

By inequalities (\ref{eqn - k >= 4/rho - 3}), (\ref{eqn - k >= 2/rho}) and
(\ref{eqn - k >= 3}), we can set $k$ as {\em any} integer such that
\begin{equation}
k \geq \max \Bigl \{\frac4\rho - 3, \frac2\rho, 3 \Bigr\}. \nonumber
\end{equation}
Once $k$ is fixed, we can fix $h$ according to (\ref{eqn - interval of h}).

We have completed our new instance $(G', c')$ of SoftMHV.

Let $m = |E(G)|$ and $n = |V(G)|$.
Denote by $m^*$ the number of happy edges in an optimal solution to
the 3-MHE instance $(G, c)$, and by $n^*$ the number of happy vertices
in an optimal solution to the SoftMHV instance $(G', c')$.
We shall prove the following claim, which will finish the proof of the
theorem.

\begin{claim}
$m^* \geq m_0 \Longleftrightarrow n^* \geq \Delta n + (h+1)m + m_0$.
\end{claim}
\begin{proof}
($\Longrightarrow$)
Let $c^*$ be an optimal solution to instance $(G, c)$.
First we color every vertex $v \in V(G')$ such that $v$ is also in $G$
by color $c^*(v)$. For each edge $(u, v) \in E(G)$, color $x_{uv}$
by color $c^*(u)$, (actually coloring $x_{uv}$ by either $c^*(u)$ or
$c^*(v)$ is ok.) and color all vertices $y_1^{uv}$, $\cdots$, $y_h^{uv}$
by the color of $x_{uv}$. This is our vertex coloring for instance $(G', c')$.

By similar arguments before inequality
(\ref{eqn - k >= 2/rho}), for every vertex $v \in V(G) \cap V(G')$ and its
corresponding $w$-vertices in $G'$, $v$ itself is unhappy and
there are exactly $\Delta$ happy vertices in
$\{w_{i,j}^v \colon 1 \leq i \leq k, 1 \leq j \leq \Delta\}$.
So we obtain $\Delta n$ happy vertices from all $w$-vertices in $G'$.

Let $(u, v)$ be an edge in $G$. In its corresponding
$y$-vertices $\{y_1^{uv}, \cdots, y_h^{uv}\}$ and
$z$-vertices $\{z_1^{uv}, \cdots, z_k^{uv}\}$,
there are exactly $h + 1$ vertices that are happy by our coloring.
So we obtain $(h + 1) m$ happy vertices from all the $y$-vertices
and $z$-vertices in $G'$.

Next let us consider vertex $x_{uv}$. If $(u, v)$ is happy by $c^*$,
then $x_{uv}$ has $h + 3$ neighbors having the same color as that of $x_{uv}$.
So the fraction of happy neighbors of $x_{uv}$ is
\begin{equation}
\frac{h+3}{h+k+2}
\geq
\frac{\frac{\rho k + 2\rho -3}{1 - \rho}+3}{\frac{\rho k + 2\rho -3}{1 - \rho} + k + 2}
= \rho, \nonumber
\end{equation}
where the first inequality is due to $h \geq \frac{\rho k + 2\rho -3}{1 - \rho}$
(by inequality (\ref{eqn - h + 3 >= rho(h + k + 2)})),
and hence $x_{uv}$ is happy.
Since $m^* \geq m_0$, we can obtain $\geq m_0$ happy vertices from
all $x$-vertices in $G'$.

Summing all, the number of happy vertices in $G'$ by our coloring
is at least $\Delta n + (h+1)m + m_0$.

($\Longleftarrow$)
Let $c'^*$ be an optimal solution to instance $(G', c')$ of SoftMHV.
By the arguments before inequality (\ref{eqn - k >= 2/rho}),
every vertex $v \in V(G') \cap V(G)$ is unhappy by $c'^*$,
and there are exactly $\Delta n$ happy $w$-vertices by $c'^*$.

Let $(u, v)$ be any edge in $G$. Since $c'^*$ is an optimal coloring,
we can assume that all vertices $y_1^{uv}$, $\cdots$, $y_h^{uv}$ have
color $c'^*(x_{uv})$. Taking into account the one more happy vertex
$z_i^{uv}$ (where $i = c'^*(x_{uv})$) for each $(u, v) \in E(G)$,
there are exactly $(h + 1)m$ happy vertices by $c'^*$ from all
$y$-vertices and $z$-vertices.

Now only $x$-vertices in $G'$ remain unconsidered.
Since $n^* \geq \Delta n + (h+1)m + m_0$, there must be at least $m_0$
happy $x$-vertices. Let $x_{uv}$ be any such vertex.
Since $x_{uv}$ is happy, the number of neighbors of $x_{uv}$
that have the color as that of $x_{uv}$ is at least
\begin{equation}
\rho (h + k + 2)
> \rho \Bigl( h + \frac{h - \rho h - 2 \rho + 2}{\rho} + 2 \Bigr)
= h + 2, \nonumber
\end{equation}
where the first inequality is due to inequality
(\ref{eqn - h + 2 < rho (h + k + 2)}).
This shows that the number of neighbors of $x_{uv}$ having color
$c'^*(x_{uv})$ is at least $h + 3$. So, vertices $u$ and $v$ must have
the {\em same} color (as that of $x_{uv}$).

Let us color every vertex $v \in V(G)$ by color $c'^*(v)$.
If there are vertices in $G$ colored by colors in $\{4, 5, \cdots, k\}$,
then color all of them by color 1 (note that $G$ is part of the instance of
the 3-MHE problem). This will never decrease the number of happy edges
in $G$. By the above analysis, the number of happy edges in $G$ is
at least $m_0$.
\qed
\end{proof}

The proof of the theorem is finished.
\qed
\end{proof}

\section{MHV with Hard Threshold}
In the hard-threshold variant of the $k$-MHV problem (HardMHV for short),
a vertex $v$ is happy if $|N^s(v)| \geq q$, where $q$ is an input
parameter.
Given a connected undirected graph $G$, a partial coloring function $c$,
and an integer $q > 0$, the HardMHV problem asks for a total vertex coloring
extended from $c$ that maximizes the number of happy vertices.
It is reasonable to assume $q \leq \Delta$, since otherwise
there is no feasible solution to the problem.

\subsection{Algorithm for HardMHV}
The following type definition of vertices is similar to
Definition \ref{def - types of vertex - SoftMHV}.

\begin{definition}[Types of vertex in HardMHV]
Fix a (partialor total) vertex coloring. Let $v$ be a vertex. Then,
\begin{enumerate}
\item $v$ is an $H$-vertex if $v$ is colored and happy,

\item $v$ is a $U$-vertex if
\begin{enumerate}
\item $v$ is colored, and
\item $v$ is destined to be unhappy (i.e., $\deg(v) - |N^d(v)| < q$),
\end{enumerate}

\item $v$ is a $P$-vertex if
\begin{enumerate}
\item $v$ is colored,
\item $v$ has not been happy (that is, $|N^s(v)| < q$), and
\item $v$ can become happy (i.e., $|N^s(v)| + |N^u(v)| \geq q$),
\end{enumerate}

\item $v$ is an $L$-vertex if $v$ has not been colored.
\end{enumerate}
\end{definition}

Similar as the case of SoftMHV, Algorithm {\sc Greedy-MHV} is also
a $1/k$-approximation algorithm for the HardMHV problem. To prove this
we only need to define $L_P$ in Theorem \ref{th - 1/k-approximation for MHV}
as the set of uncolored vertices $v$ such that
$|N^u(v)| + \max \{|N_i(v)|\} \geq q$, and $L_D = L - L_P$.

\begin{theorem}
There is a $1/k$-approximation algorithm for the HardMHV problem. \qed
\end{theorem}

In the MHV and SoftMHV problems, for an $L$-vertex $v$, if $|N^d (v)|$ is
too large, then $v$ may be destined to be unhappy.
In contrast, in the HardMHV problem, an $L$-vertex $v$ may be destined to
be unhappy even if $|N^d (v)| = 0$: This will happen when $\deg(v) < q$.
Based on this observation, the $L$-vertex type is divided into the following
four subtypes.

\begin{definition}[Subtypes of $L$-vertex in HardMHV]
\label{def - subtypes of L-vertex - HardMHV}
Let vertex $v$ be an $L$-vertex in a vertex coloring. Then,
\begin{enumerate}
\item $v$ is an $L_p$-vertex if $v$ is adjacent to a $P$-vertex,

\item $v$ is an $L_{h}$-vertex if
\begin{enumerate}
\item $v$ is not adjacent to any $P$-vertex,
\item $v$ is adjacent to an $H$-vertex or a $U$-vertex, and
\item $v$ can become happy (i.e.,
$|N^u(v)| + \max \{|N_i(v)| \colon 1 \leq i \leq k \} \geq q$),
\end{enumerate}

\item $v$ is an $L_{u}$-vertex if
\begin{enumerate}
\item $v$ is not adjacent to any $P$-vertex, and
\item $v$ is destined to be unhappy (i.e.,
$|N^u(v)| + \max \{|N_i(v)| \colon 1 \leq i \leq k \} < q$),
\end{enumerate}

\item $v$ is an $L_f$-vertex if
\begin{enumerate}
\item $v$ is not adjacent to any colored vertex, and
\item $v$ can become happy.
\end{enumerate}
\end{enumerate}
\end{definition}

One can verify that the subtypes in
Definition \ref{def - subtypes of L-vertex - HardMHV}
really form a partition of all $L$-vertices. Note that the $L_{u}$-vertex
not only refers to the destined-to-be-unhappy $L$-vertex that is adjacent
to an $H$-vertex or a $U$-vertex (like the $L_{u}$-vertex in MHV and
the $L_{u}$-vertex in SoftMHV), but also refers to the
destined-to-be-unhappy $L$-vertex that is not adjacent to any colored vertex,
as discussed before Definition \ref{def - subtypes of L-vertex - HardMHV}.

Below is the subset-growth approximation algorithm {\sc Growth-HardMHV}
for the HardMHV problem.

\setcounter{algleo}{0}
\begin{algleo}
\linonumber {\bf Algorithm} {\sc Growth-HardMHV}
\linonumber {\em Input:} A connected undirected graph $G$, a partial
    coloring function $c$, and an integer $q > 0$.
\linonumber {\em Output:} A total vertex coloring for $G$.

\li $\forall 1 \leq i \leq k$, $V_i \leftarrow \{v \colon c(v) = i\}$.
\li {\bf while} there exist $L$-vertices {\bf do}
\begin{algleo}
    \li {\bf if} there exists a $P$-vertex $v$ {\bf then}
        \begin{algleo}
            \li $i \leftarrow c(v)$.
            \li Add its any $q - |N^s(v) \cap V_i|$ $L_p$-neighbors to $V_i$.
                The types of all affected vertices (including $v$ and vertices
                in $N^2(v)$) are changed accordingly.
        \end{algleo}
    \li {\bf elseif} there exists an $L_{h}$-vertex $v$ {\bf then}
        \begin{algleo}
            \li Let $V_i$ be the vertex subset in which $v$ has the maximum
                colored neighbors.
            \li Add vertex $v$ and its any $q - |N^s(v) \cap V_i|$ $L$-neighbors
                to $V_i$.
                The types of all affected vertices (including $v$ and vertices
                in $N^2(v)$) are changed accordingly.
        \end{algleo}
    \li {\bf else}
        \begin{algleo}
            \linonumber {\em Comment:} There must be an $L_{u}$-vertex.
            \li Let $v$ be any $L_{u}$-vertex. If $v$ has colored neighbors,
                then let $V_i$ be any vertex subset containing a colored neighbor of $v$.
                Otherwise let $V_i$ be $V_1$.
            \li Add vertex $v$ to subset $V_i$.
                The types of all affected vertices (including $v$ and vertices in $N(v)$)
                are changed accordingly.
        \end{algleo}
    \li {\bf endif}
\end{algleo}
\li {\bf endwhile}
\li {\bf return} the vertex coloring $(V_1, V_2, \cdots, V_k)$.
\end{algleo}

\begin{lemma}
\label{lm - lower bound on |H^new|, Growth-HardMHV}
$|L_{u}^{new}| \leq O(\Delta^2) |H^{new}|$.
\end{lemma}
\begin{proof}
The proof of the lemma is similar to that of
Lemma \ref{lm - lower bound on |H^new|, Growth-SoftMHV}.
Only one point needs to pay attention.
When the algorithm processes an $L_{u}$-vertex, there are only
$L_{u}$-vertices or $L_f$-vertices (if any) in the current graph.
Each time Algorithm {\sc Growth-HardMHV} processes an $L_{u}$-vertex $v$,
it processes only {\em one} such vertex. So, if $v$ has an $L_f$-neighbor
$u$, $u$ will become an $L_{h}$-vertex after the processing.
This means that coloring an $L_{u}$-vertex does not generate any
new $L_{u}$-vertex. We omit the other details of the proof.
\qed
\end{proof}

\begin{theorem}
The HardMHV problem can be approximated within a factor of
$\Omega(\Delta^{-3})$ in polynomial time.
\end{theorem}
\begin{proof}
Each time an $H$-vertex is generated, at most $q$ $L$-vertices are consumed
(i.e., colored). So, for the number of newly generated $H$-vertices we have
$|H^{new}| \geq (|L^{org}| - |L_{u}^{org}| - |L_{u}^{new}|) / q$.
By Lemma \ref{lm - lower bound on |H^new|, Growth-HardMHV}, and noticing
that $q \leq \Delta$, we get
\begin{equation}
|H^{new}| \geq \frac{|L^{org}| - |L_{u}^{org}|}{O(\Delta^2)}.
\nonumber
\end{equation}

Let $OPT$ be the number of happy vertices in an optimal solution to
the problem. By the same reason as in Lemma \ref{lm - upper bound on OPT},
we obtain
\begin{equation}
OPT \leq |H^{org}| + (\Delta+1) (|L^{org}| - |L_{u}^{org}|). \nonumber
\end{equation}

Let $SOL$ be the number of happy vertices found by Algorithm
{\sc Growth-HardMHV}. Then we have
$SOL = |H^{org}| + |H^{new}| = \Omega(\Delta^{-3}) OPT$
by the above two inequalities.
As Algorithm {\sc Growth-HardMHV} obviously runs in polynomial time,
the theorem follows.
\qed
\end{proof}

\subsection{NP-Hardness of HardMHV}
\begin{theorem}
The HardMHV problem is NP-hard for any constant $k \geq 3$, where $k$ is
the color number in the problem.
\end{theorem}
\begin{proof}
We prove the theorem by reducing $k$-MHE (see Theorem \ref{th - k-MHE is NP-hard})
to HardMHV.

Given an instance $(G, c)$ of $k$-MHE, we construct an instance $(G', c', q)$
of HardMHV as follows. For each edge $(u, v) \in E(G)$, do the following.
Add a vertex $x_{uv}$ and $\Delta - 1$ vertices
$y_1^{uv}$, $y_2^{uv}$, $\cdots$, $y_{\Delta}^{uv}$, where $\Delta$ is
the maximum vertex degree of $G$. The vertices $y_i^{uv}$'s are called
{\em satellite vertices}.
Replace edge $(u, v)$ by two edges $(u, x_{uv})$ and $(x_{uv}, v)$.
Connect each vertex $y_i^{uv}$ to $x_{uv}$ via an edge $(x_{uv}, y_i^{uv})$.
Finally, let $q = \Delta + 1$. We thus get our HardMHV instance
$(G', c', q)$.

Since $q = \Delta + 1$, each original vertex $v \in V(G)$ and each newly
added satellite vertex cannot be happy no matter how the vertices in $G'$
are colored. For each edge $(u, v) \in E(G)$, since its corresponding vertex
$x_{uv}$ is of degree $\Delta + 1$, $x_{uv}$ is happy iff its
two neighbors $u$ and $v$ have the same color. This shows that
the optimum of $(G, c)$ is equal to that of $(G', c', q)$, finishing
the proof of the theorem.
\qed
\end{proof}

\end{document}